\documentclass[runningheads]{llncs}
\usepackage[T1]{fontenc}
\usepackage{graphicx}
\usepackage{url}
\usepackage{graphicx}
\usepackage{tcolorbox}
\AtBeginDocument{%
  \providecommand\BibTeX{{%
    \normalfont B\kern-0.5em{\scshape i\kern-0.25em b}\kern-0.8em\TeX}}}

\usepackage{graphicx}
\usepackage{wrapfig}
\usepackage{booktabs} 

\usepackage{url}
\usepackage{algorithm}
\usepackage[noend]{algpseudocode}
\usepackage{amssymb,amsmath}
\usepackage{subcaption}
\usepackage{graphicx}
\usepackage{comment}
\usepackage{thmtools, thm-restate}
\usepackage{multirow}

\usepackage[graphicx]{realboxes}

\usepackage{xcolor} 
\definecolor{ocre}{RGB}{243,102,25} 

\usepackage{avant} 
\usepackage{mathptmx} 

\usepackage{microtype} 
\usepackage[utf8x]{inputenc} 

\usepackage{tikz}
\usetikzlibrary{shapes}
\RequirePackage{ifthen}
\RequirePackage{amsmath}
\RequirePackage{amssymb}
\RequirePackage{xspace}
\RequirePackage{graphics}
\usepackage{color}
\usepackage{keyval}
\usepackage{xspace}
\usepackage{mathrsfs}

\newcommand{\reals}{{\mathbb{R}}} 

\newcommand{\nnreals}{{\mathbb{R}_{\geq 0}}}
\newcommand{\naturals}{{\mathbb{N}}} 
\newcommand{\dom}{\relax\ifmmode {\mathit{dom}} \else ${\sf dom}$\fi} 

\newcommand{\redreachset}{\mathit{ReachDict}}

\newcommand{\unsafeset}{{\mathit{X_a}}} 

\newcommand{\rootnode}{\mathit{root}} 
\newcommand{\traversalstack}{\mathit{stack}} 
\newcommand{\node}{\mathit{node}}
\renewcommand{\path}{\mathit{path}}

\newcommand{\reachset}{\mathit{Reachset}} 
\newcommand{\overreachset}{\mathit{ApprReachset}}

\newcommand{\state}{{x}}

\newcommand{\initset}{\Theta}
\newcommand{\redinitset}{\bar{\Theta}}

\newcommand{\stateset}{X}
\newcommand{\discretestateset}{\mathcal{X}}
\newcommand{\discretestateinst}{s}
\newcommand{\reddiscretestateset}{\bar{\discretestateset}}
\newcommand{\redstateset}{\bar{\stateset}}
\newcommand{\redgrid}{\bar{Q}}
\newcommand{\abstractunder}{\mathit{sym}}
\newcommand{\gridunder}{\mathit{grid}}
\newcommand{\redstateinst}{\bar{\stateinst}}
\newcommand{\reddiscretestateinst}{\bar{s}}
\newcommand{\reddisstateset}{\bar{\disstateset}}
\newcommand{\statedim}{n_x}
\newcommand{\redstatedim}{\bar{n}_x}
\newcommand{\inputdim}{n_u}
\newcommand{\redinputdim}{\bar{n}_u}
\newcommand{\distdim}{n_w}
\newcommand{\inputmap}{\beta}
\newcommand{\timestep}{\tau}
\newcommand{\constantinputsignals}{\mathcal{U}_c}
\newcommand{\reddistdim}{\bar{n}_w}
\newcommand{\undersym}{_{\mathit{rel}}}
\newcommand{\inputset}{U}
\newcommand{\redinputset}{\bar{\inputset}}
\newcommand{\redinputinst}{\bar{\inputinst}}
\newcommand{\disinputset}{\mathcal{U}}
\newcommand{\disinputinst}{a}
\newcommand{\reddisinputset}{\bar{\disinputset}}
\newcommand{\distset}{W}
\newcommand{\distinst}{w}
\newcommand{\reddistinst}{\bar{w}}
\newcommand{\reddistset}{\bar{\distset}}
\newcommand{\reddynamics}{\bar{f}}
\newcommand{\reddelta}{\bar{\delta}}
\newcommand{\symcrosssection}{\mathcal{C}}
\newcommand{\liegroup}{\mathcal{G}}

\newcommand{\stateinst}{{\bf x}}
\newcommand{\inputinst}{u}

\newcommand{\fsr}{\mathcal{R}}

\newcommand{\rv}{\mathcal{R}_{gs}}

\newcommand{\cache}{\mathit{cache}}

\newcommand{\parent}{\mathit{parent}}

\newcommand{\last}{\mathit{last}}

\newcommand{\pop}{\mathit{pop}}

\newcommand{\result}{\mathit{result}}

\newcommand{\agent}{\mathcal{A}}

\newcommand{\controller}{\mathcal{K}}
\newcommand{\redcontroller}{\bar{\mathcal{K}}}
\newcommand{\retrieve}{\mathit{retrieve}}
\newcommand{\insertK}{\mathit{insert}}

\newcommand{\redConcState}{\mathit{findFrameContainX}}
\newcommand{\redConcStateInter}{\mathit{findFrameInterX}}
\newcommand{\redConcInput}{\mathit{findFrameU}}
\newcommand{\applyInvFrame}{\mathit{applyInvFrame}}
\newcommand{\fullreach}{\mathit{full}}
\newcommand{\lastreach}{\mathit{last}}

\usepackage{setspace}
\usepackage{listings}

\usepackage[normalem]{ulem}

\usepackage{placeins}
\makeatletter
\newcommand\nocaption{%
    \renewcommand\p@subfigure{}
    \renewcommand\thesubfigure{\thefigure\alph{subfigure}}
}
\makeatother

\title{\LARGE \bf
	Symmetry-based Abstraction Algorithm for Accelerating Symbolic Control Synthesis 
}

\author{Hussein Sibai\inst{1} \and Sacha Huriot\inst{1} \and Tyler Martin\inst{1} \and Murat Arcak\inst{2} 
\\  
           \institute{}
		       \inst{1}  Washington University in St. Louis \\
		         {\tt\small \{sibai,h.sacha,martin.t\}@wustl.edu} \\
  \inst{2} University of California, Berkeley \\ 
{\tt\small arcak@berkeley.edu }}
 
\begin{document}

	\maketitle
	\thispagestyle{empty}
	\pagestyle{empty}

	\begin{abstract}
 We propose an efficient symbolic control synthesis algorithm for equivariant continuous-time dynamical systems to satisfy reach-avoid specifications. The algorithm exploits dynamical symmetries to construct lean abstractions to avoid redundant computations during synthesis.
	Our proposed algorithm adds another layer of abstraction over the common grid-based discrete abstraction  before solving the synthesis problem.
	It combines each set of grid cells 
	that are at a similar relative position from the targets and nearby obstacles, defined by the symmetries, into a single abstract state. It uses this layer of abstraction to guide the order by which actions are explored during synthesis over the grid-based abstraction. 
 We demonstrate the potential of our algorithm by synthesizing a reach-avoid controller for a 3-dimensional ship model with translation and rotation symmetries in the special Euclidean group SE(2)\footnote{Code is available at: \url{https://github.com/HusseinSibai/symmetric_control_synthesis/tree/public
}.}.
	\end{abstract}

	\section{INTRODUCTION}
	
 Control synthesis is the problem  of automatically generating a 
	control signal that drives a dynamical system to satisfy a given specification.  If solved, it offers a correct-by-construction approach to formally assuring that control systems behave as intended. Several promising approaches for synthesis have been proposed in the past for discrete and continuous-time dynamical systems and for reach-avoid and optimality specifications~\cite{Meyer2020_control_synthesis_ship_docking,Reissig_feeback_refinement_relations_synthesis_tac_2016,Girard_heirarchical_control_automatica_2009,tabuada_2009,Pola_decenteralized_control_network_2018,symmetry_reduction_dynamic_programming}.  The main challenge for 
 wider deployability 
	is scalability, as these
 algorithms require expensive computational time and memory. They are known to suffer from what is colloquially known as the curse-of-dimensionality since their computational complexity grows at least exponentially with the state and input dimensions~\cite{tabuada_2009}. 
	
	On the other hand, many control systems possess symmetries. A symmetry map of a dynamical system acts on its state space. When applied on the states visited by any of the system trajectories, a symmetry map results in another trajectory of the same system  
	starting from a different initial state, potentially following different control and disturbance signals. A control synthesis algorithm explores the feasible control signals at each state and chooses a one that drives the system to satisfy the given specification. With symmetric dynamics, such an algorithm can infer the control that satisfies the specification at a given state without its usual comprehensive exploration if it knows the one that satisfies the specification at a symmetric state. For example, a vehicle with translation- and rotation-symmetric dynamics  will likely need the same control to satisfy a reach-avoid specification when starting from states at a similar {\em relative} position from the reach and avoid sets. 
    We exploit this simple, yet effective, intuition by constructing abstractions that combine such similar {\em concrete} states into the same {\em abstract} states. 
 We use these abstractions to explore the control space more efficiently during control synthesis. We further use symmetries to compute fewer reachable sets from scratch using existing tools during the construction of the abstractions, which in turn  can be efficiently transformed to obtain the rest of the reachable sets.
	
	Abstractions 
 have helped solving several formal verification and synthesis problems in various domains, e.g., \cite{Doyen_AR_2005,Assume-Guarantee-AR-Sergiy-2014,HARE_2017,TIOAmon,Meyer2020_control_synthesis_ship_docking,sibai2021scenechecker}. 
	Common abstractions for continuous-time dynamical systems are discrete transition models, e.g.,~\cite{tabuada_2009,Meyer2020_control_synthesis_ship_docking}. 
	Such discrete models are usually built by gridding the state, control, and time spaces. We call them {\em grid-based abstractions} (GA). Then, their non-deterministic discrete transitions are constructed by computing the bounded-time reachable sets, or corresponding over-approximations, of the continuous systems starting from the cells of the grid following constant control for a constant time interval, using reachability analysis tools, e.g., \cite{TIRA_Meyer_2019,c2e2}. Such an abstraction results in discrete state and control sets with sizes exponential in their dimensions. To synthesize a reach-avoid controller for such an abstraction, the standard algorithm repetitively iterates, with an {\em arbitrary} order, over these 
 sets. At each iteration, the algorithm extends the target or {\em reach} set with the discrete states that can reach it in one transition without intersecting with the {\em avoid} set~\cite{tabuada_2009,Meyer2020_control_synthesis_ship_docking}. 
 Each iteration of the algorithm has an exponential time complexity in the state and input dimensions.

	
	
 We further abstract the continuous-time system by combining the cells in the grid that are at {\em similar} relative positions to the reach and avoid sets. We call the new system the {\em symmetry-based abstraction} (SA). 
 We use the SA to define an order over which different controls are explored during synthesis at every  GA state. Such an order allows more efficient exploration of the action space during synthesis. The order will reflect the likelihood of control symbols satisfying the specification at a GA state given the experience at other GA states 
 that are represented by the same SA state. 

 Moreover, we use symmetries to compute 
 a set of reachable sets with cardinality equal to the number of cells in the grid over the control space times the number of cells in the non-symmetric dimensions of the state space. That is exponentially fewer than the ones used to compute the GA transitions in the traditional method (e.g.,~\cite{Meyer2020_control_synthesis_ship_docking}). Also, each of the reachable sets we compute starts from an initial set of states that has uncertainty only in the non-symmetric dimensions of the state. These reachable sets are represented in the relative coordinates by construction and are sufficient to construct the SA. When a GA transition is needed during synthesis, it is sufficient to transform one of these pre-computed reachable sets using particular symmetry transformations (Corollary~\ref{cor:lower_dimensional_reachable_sets}). Such transformations are potentially more efficient than computing the reachable set from scratch using existing reachability analysis tools~\cite{sibai-atva-2019,sibai-tacas-2020,sibai2021scenechecker}.

 The algorithm for constructing SA transforms each cell in the grid to a manifold where the symmetric coordinates are constant. Then, it finds the control symbol that results in the closest reachable set to the target that does not intersect the obstacles, when the cell is the initial set of states. All sets, i.e., the initial, reachable, target, and obstacles' ones, are represented in the relative coordinates based on the original concrete coordinates of the cell. The SA construction algorithm then maps the cells which share the same such control to the same SA state. 
 
 Before synthesis, we use the constructed SA to   
 build a cache 
 mapping every SA state to a set of pairs of non-negative integers, called {\em scores}, and control symbols. 
 At first, every entry is initialized to a singleton set  with the pair having the control symbol that was used to construct the corresponding SA state and its score being zero. During control synthesis for GA, this cache is used to 
define the order according to which control symbols are checked at every GA state.
Particularly, if at a certain synthesis iteration a control symbol is found to be specification-satisfying for some GA state, the cache entry corresponding to its representative SA state is updated. Specifically, if that symbol exists in that cache entry, its score is incremented by one. Otherwise, a new pair is added to the entry with a score of one.
The synthesis algorithm explores the symbols at a given GA state in a decreasing score order checking the controls that satisfied the specification for more GA states with the same representative SA state first.
Our algorithm preserves the guarantees of the traditional synthesis algorithm for GA. We present experimental results showing promising computational time savings when synthesizing a reach-avoid controller for a 3-dimensional ship model while exploiting its translation and rotation symmetries.

\paragraph{Related work}
\label{sec:appendix_relatedwork}
Symmetries have been useful in numerous domains: from deriving  conservation laws~\cite{Noether1918}, analyzing contraction and stability~\cite{russo2011symmetries}, reinforcement and supervised learning algorithms~\cite{equivariant_reinforcement_learning_vanderpol_neurips_2020,group_equivariant_nn}, and interval integration of ODEs~\cite{guaranteed_integration}.

Symmetries have been used to efficiently design motion planners for vehicles to achieve reach-avoid specifications~\cite{factest_cav_2020,Majumdar2017}. Particularly, they have been used to transform pre-computed reachable sets to obtain reachable sets starting from different initial sets.
The plans these works synthesize start from a fixed initial set of states and assume given user-provided feedback controllers. Symmetries have also been used to reduce the dimensionality of dynamic programming problems when both the system dynamics and the cost function share the same symmetries~\cite{symmetry_reduction_dynamic_programming}.
	In contrast, we do not assume a given feedback controller, and do not assume that the specification has the same symmetries as the dynamics. We aim to find the largest initial set in the state set from which the reach-avoid specification can be reached along with the associated controller. 

Symmetries have also been useful to accelerate  safety verification of continuous-time dynamical systems and hybrid automata~\cite{MaidensSym,sibai-atva-2019,sibai-tacas-2020,sibai2021scenechecker}. More specifically, they have been used to reduce the dimensionality of reachable set computation problems, to efficiently cache reachable sets during verification, and to generate abstractions of hybrid automata. 
 These works tackle the verification problem for closed dynamical systems, while we tackle the control synthesis problem for systems with inputs and disturbances.
	\section{Preliminaries}

	\paragraph{Notation} We denote by $\naturals$, $\reals$ and $\nnreals$ the sets  of natural, real, and non-negative real numbers, respectively. 
	$\forall n \in \naturals$, we define $[n] = \{0,\dots, n-1\}$.  
	For any $g: S \times \nnreals \rightarrow S$, 
	we denote by $g(s, \cdot): \nnreals \rightarrow S$ the function of $\nnreals$ resulting from  projecting $g$ to having the first parameter fixed to $s \in S$. If $S$ is a finite set, we denote its cardinality by $|S|$. 
	\vspace{-0.1in}
	\subsection{System and problem definitions}
	We consider control systems described with ordinary differential equations (ODEs): \begin{align}\label{eq:sys}
		\dot{x}(t) = f(\state(t), \inputinst(t), \distinst(t)),
	\end{align}
	where $\state(t) \in \reals^{n_\state}$ is the {\em state} of the system at time $t\in \nnreals$,  $\inputinst: \nnreals \rightarrow \inputset \subset \reals^{n_\inputinst}$  and $\distinst: \nnreals \rightarrow \distset \subseteq \reals^{n_\distinst}$ are two measurable functions representing control and disturbance inputs with $\inputset$ and $\distset$ being compact sets,
	and $f: \reals^{n_\state} \times \reals^{n_\inputinst} \times \reals^{n_\distinst} \rightarrow \reals^{n_\state}$ is a Lipschitz continuous function in $\state$ uniformly in $\inputinst$ and $\distinst$. 
	Given any time bound $T \geq 0$, initial state $\stateinst_0 \in \reals^{n_\state}$, input signal $\inputinst: [0,T] \rightarrow \reals^{n_\inputinst}$, and a disturbance signal $\distinst: [0,T] \rightarrow \reals^{n_\distinst}$, we assume that the corresponding {\em trajectory} of system~(\ref{eq:sys}) exists and is unique. 
	We denote such a  trajectory by: $\xi(\stateinst_0, \inputinst, \distinst; \cdot): [0, T] \rightarrow \reals^{n_\state}$, and say that $T$ is the duration of $\xi$. The trajectory $\xi$ satisfies  equation~(\ref{eq:sys}) almost everywhere and is equal to $\stateinst_0$ at $t = 0$. 
	
	Given a time interval $[t_0, t_1] \subset \nnreals$, a compact initial set $\stateset_0 \subseteq \stateset$, a fixed control function $\inputinst: \nnreals \rightarrow \reals^{n_\inputinst}$, and the interval disturbance set $\distset$, the {\em reachable set} of system~(\ref{eq:sys})  is defined as follows:
		$ \reachset(\stateset_0, \inputinst, \distset; [t_0,t_1]) := \{\xi(\stateinst_0, \inputinst, \distinst; t) \ |\ 
  \stateinst_0 \in \stateset_0, t\in [t_0,t_1], \distinst: \nnreals \rightarrow \distset \}$,
	where every $\distinst$ in the equation above is a measurable signal. We abuse notation and denote   the set of states that are reachable exactly at time $t$, i.e., $\reachset(\stateset_0, \inputinst, \distset; [t, t])$, by $\reachset(\stateset_0, \inputinst, \distset; t)$. 
	It is generally undecidable to compute the exact reachable sets of nonlinear control systems~\cite{henz:thhybat}. Fortunately, significant progress resulted in numerous tools to over-approximate them~(e.g., \cite{TIRA_Meyer_2019,c2e2}).



	Given {\em reach} and {\em avoid} sets $\stateset_r, \stateset_a \subseteq \reals^{n_\state}$, a {\em reach-avoid} 
	specification of system~(\ref{eq:sys}) requires its trajectories to eventually reach $X_r$, while avoiding being in $X_a$. The ``reach set'' $\stateset_r$ should not be confused with the ``reachset'', or equivalently, reachable set, that we defined earlier.
 Now, we formally define the problem as follows:
	
	\begin{problem}
		\label{def:problem}
		Given system~(\ref{eq:sys}) and  
		{\em reach} and {\em avoid} sets of states $\stateset_r, \stateset_a \subseteq \stateset$, 
		find a set of initial states $\stateset_0 \subseteq \stateset$, preferably the largest one possible, and a controller $g: \stateset \times \nnreals \rightarrow \inputset$ such that for any measurable disturbance function $\distinst: \nnreals  \rightarrow \distset$, $\stateinst_0 \in \stateset_0$, $\exists t_r \in \nnreals$,
		such that $\forall t \leq t_r$, $\xi(\stateinst_0, g(\stateinst_0,\cdot),\distinst; t) \in \stateset$\textbackslash$\stateset_a$ 
		and
		$\xi(\stateinst_0, g(\stateinst_0,\cdot),\distinst; t_r) \in \stateset_r$. 
	\end{problem}
	
	
	\subsection{Transformation groups and equivariant control systems  }

	A control system whose dynamics are unchanged under transformations of its state, control, and disturbance 
	is called {\em equivariant}. Formally,  fix a {\em transformation group} $\{h_\alpha = (\phi_\alpha, \chi_\alpha, \psi_\alpha)\}_{\alpha \in \liegroup}$ on $\reals^{n_\state} \times \reals^{n_\inputinst} \times \reals^{n_\distinst}$,
 where $\liegroup$ is a Lie group with dimension $r$. 
	Then,	system~(\ref{eq:sys}) is called $\liegroup$-{\em equivariant} if $\forall \alpha \in \liegroup$, $\stateinst \in \reals^{n_\state}$, $\inputinst \in  \reals^{n_\inputinst}$, and $\distinst \in  \reals^{n_\distinst}$, it satisfies
		$\frac{\partial \phi_{\alpha} }{\partial \state }|_{\state = \stateinst} f(\stateinst, \inputinst, \distinst) = f(\phi_{\alpha}(\stateinst), \chi_\alpha(\inputinst), \psi_\alpha(\distinst))$~\cite{symmetry_reduction_dynamic_programming}.
	A trajectory or a reachset of a $\liegroup$-equivariant system can be transformed using the corresponding transformation group to obtain more trajectories or reachsets, as shown in the following theorems.
	
	\begin{theorem}[\cite{russo2011symmetries}]
		\label{thm:symmetric_trajectories}
		If system~(\ref{eq:sys}) is $\liegroup$-equivariant, then $\forall \alpha \in \liegroup$, $\stateinst_0 \in \reals^{n_\state}$, $\inputinst: \nnreals \rightarrow \reals^{n_\inputinst}$, $\distinst: \nnreals \rightarrow \reals^{n_\distinst}$, and  $t \in \nnreals$, 
		$\phi_\alpha(\xi(\stateinst_0,\inputinst, \distinst;t)) = \xi(\phi_\alpha(\stateinst_0),\chi_\alpha \circ \inputinst, \psi_\alpha \circ \distinst;t)$.
	\end{theorem}

	
	\begin{theorem}[\cite{sibai-atva-2019,sibai-tacas-2020}]
		\label{thm:symmetric_reachsets}
		If system~(\ref{eq:sys}) is $\liegroup$-equivariant, then $\forall \alpha \in \liegroup$, $\stateset_0 \subset \reals^{n_\state}$, $\inputinst: \nnreals \rightarrow \reals^{n_\inputinst}$, $\distinst: \nnreals \rightarrow \reals^{n_\distinst}$, and  $[t_0,t_1] \subset \nnreals$, 
		$ \phi_\alpha(\reachset(\stateset_0,\inputinst, \distset;[t_0,t_1])) =$
   $\reachset($ $\phi_\alpha(\stateset_0),\chi_\alpha \circ \inputinst, \psi_\alpha ( \distset);[t_0,t_1])$.
		Moreover, if we let $\overreachset(\stateset_0,\inputinst, \distinst;[t_0,t_1])$ be a set in $ \reals^{n_\state}$ over-approximating the exact reachable set $\reachset(\stateset_0,\inputinst, \distinst;[t_0,t_1])$, 
  Then, $\reachset(\phi_\alpha(\stateset_0),\chi_\alpha \circ \inputinst, \psi_\alpha(\distset);[t_0,t_1]) \subseteq 
   \phi_\alpha(\overreachset(\stateset_0,\inputinst, \distset;[t_0,t_1]))$.
	\end{theorem}
	
	In the rest of the paper, we use $\reachset$ for both the exact values and the over-approximations of reachable sets. We clarify when a distinction has to be made. 

	\subsection{Cartan's moving frame and computing lower-dimensional reachsets}
	\label{sec:cartan}
	
	
	
	Assume that system~(\ref{eq:sys}) is $\liegroup$-equivariant and that $\liegroup$ is $r$-dimensional. Moreover, assume that $\forall \alpha \in \liegroup$, $\phi_\alpha$ can be split as $(\phi_\alpha^a, \phi_\alpha^b)$, defining the $r$ and $n_\state- r$ dimensions of the image of $\phi_\alpha$, respectively, such that $\phi_\alpha^a$ is invertible. We call the first $r$ dimensions the {\em symmetric coordinates} and the rest the {\em non-symmetric} ones. We select a $c$ in the image of $\phi_{\alpha}^a$ and define the cross-section $\symcrosssection = \{\stateinst\ |\ \phi_{e}^a(\stateinst) = c\}$, a submanifold of $\reals^{\statedim}$. We assume that for each $\stateinst \in \reals^{\statedim}$, there exists a unique $\alpha \in \liegroup$ such that $\phi_{\alpha}(\stateinst) \in \symcrosssection$. We denote by $\gamma: \stateset \rightarrow \liegroup$ the function  that maps each $\stateinst \in \stateset$ to  its corresponding $\alpha \in \liegroup$. We call $\gamma$ a {\em moving frame}. This concept has been proposed by Cartan in 1937~\cite{Cartan1937LaTD} and has found numerous applications (e.g.,~\cite{symmetry_preserving_observers,jakubczyk1998symmetries,MaidensSym,sibai2021scenechecker,sibai-tacas-2020,sibai-tac-2020}).
	
	In the following corollary of Theorem~\ref{thm:symmetric_reachsets}, we show that when a moving frame is available, we can decompose the computation of a reachable set of system~(\ref{eq:sys}) starting from an initial set $\stateset_0$ into the computation of another reachable set starting from the projection of $\stateset_0$ to the cross-section $\symcrosssection$ followed by a set of $r$-dimensional transformations.
	A version of the corollary for translation symmetry has been proposed in~\cite{sibai-atva-2019}. Also, a similar idea has been utilized for interval integration using general Lie symmetries in~\cite{guaranteed_integration}. We present the proof of the corollary in Appendix~\ref{sec:appendix_proofs}.
 \begin{restatable}[]{corollary}{symmetryreachablesets}
	\label{cor:lower_dimensional_reachable_sets}
	If system~(\ref{eq:sys}) is $\liegroup$-equivariant where $\liegroup$ is $r$-dimensional, $\gamma$ is a corresponding moving frame, and $\forall \alpha \in \liegroup$,   $\phi_{\alpha}^b$ is the map that projects $\state$ to its last $\statedim - r$ coordinates and $\chi_{\alpha}$ is the identity map, then $\forall\stateset_0 \subset \reals^{n_\state}$, $\inputinst: \nnreals \rightarrow \reals^{n_\inputinst}$, 
 and  $[t_0,t_1] \subset \nnreals$: 
		 $\reachset(\stateset_0, \inputinst, \distset;[t_0,t_1]) 
   \subseteq \cup_{\stateinst_0 \in \stateset_0}{\phi_{\gamma(\stateinst_0)}}^{-1}(\reachset(\redstateset_0,\inputinst, \reddistset;[t_0,t_1]))$,
		where $\redstateset_0 = \cup_{\stateinst_0 \in \stateset_0}\phi_{\gamma(\stateinst_0)}(\stateinst_0)$ and $\reddistset = \cup_{\stateinst_0 \in \stateset_0}\psi_{\gamma(\stateinst_0)}(\distset)$.
 \end{restatable}
	
	The set $\redstateset_0$ in the corollary above belongs to the cross-section $\symcrosssection$ determined by the moving frame $\gamma$. Thus, the projection of any $\redstateinst_0 \in \redstateset_0$ to its first $r$ dimensions is equal to $c$. Computing the reachable set starting from $\redstateset_0$, then transforming it using the maps $\phi_{\alpha}$ is usually computationally cheaper than computing it starting from $\stateset_0$~\cite{sibai-atva-2019}. For instance, many reachability analysis tools, e.g.,~\cite{c2e2}, refine, i.e., partition, the initial sets to decrease over-approximation errors. Consequently, having lower-dimensional initial sets can result in exponentially-fewer partitions 
 when computing accurate reachable sets.  
 The union over $\stateset_0$ in the corollary
 can be over-approximated in a similar manner to computing reachable sets for discrete-time systems with uncertain parameters and initial states for a single time step.  
	Finally, the assumption in the corollary that $\chi_\alpha$ is an identity map is usually satisfied as the control input is mostly represented in the body coordinates in robotic systems, e.g., pushing the throttle or steering the  wheel. 

	\subsection{Discrete abstractions of continuous-time control systems}
	\label{sec:continuous_to_discrete_abstractions}
	
	In this section, we define discrete abstractions of continuous-time dynamical systems under reach-avoid specifications.
	When abstracting the transitions of continuous-time systems under reach-avoid specifications using discrete ones, it is important to distinguish the discrete states representing the continuous states that the continuous-time system {\em reaches} at the end of a time period from the ones it  {\em visits} as it evolves over the time period.
	We only need  the states that the system reaches to be inside the target set, but all of the states that it visits to be outside the avoid set. 
 Without this distinction, the initial state will always be a potential {\em next} state in the discrete transition, and unless it is already in the target set, we would not be able to add the initial state to the extended target set.  We define a transition map that makes such a distinction as follows.
	
	\begin{definition}[Non-deterministic discrete system]
	 A non-deterministic discrete system $\agent$ is a tuple $ (\discretestateset, \disinputset, \delta)$, where $\discretestateset$ and $\disinputset$ are two finite sets of symbols, and $\delta: \discretestateset \times \disinputset \rightarrow 2^{2^{\discretestateset} \times \discretestateset}$ is a non-deterministic transition map.
	\end{definition}
	
	We denote 
 $\delta$
 when restricted to the first entries of the pairs in its image by $\delta.\fullreach$, representing the sets of states that
 $\agent$ might {\em visit} in a transition, and when restricted to the second entries by $\delta.\lastreach$, representing the states that it might {\em reach} after the transition.
	
	The semantics of a discrete system $\agent$ are determined by its executions. An {\em execution} $\sigma$ of $\agent$ is a sequence of triples. The first entry of each triple is a finite set of symbols in $\discretestateset$, the second entry is a symbol contained in the first entry, and the last entry is a symbol in $\disinputset$, or is equal to $\bot$, indicating the end of the execution. More formally,
	    $\sigma := \langle (\{ \discretestateinst_{0}\}, \discretestateinst_0, \disinputinst_0), \dots, (\{ \discretestateinst_{l,j}\}_{j\in [z_l]}, \discretestateinst_l, \bot) \rangle$,
	where $\discretestateinst_0 \in \discretestateset$ is an initial state, and $\forall i \leq l$, $z_i$ is the number of states that the system might visit in the $i^{\mathit{th}}$ step. Moreover, $(\{\discretestateinst_{i+1,j}\}_{j \in [z_{i+1}]}, \discretestateinst_{i+1}) \in \delta(\discretestateinst_{i}, \disinputinst_i)$. 
 Finally, we say that $l$ is the {\em length} of $\sigma$.
	
	Given reach and avoid sets of discrete states $\discretestateset_r$ and $\discretestateset_a$, respectively, we say that an execution $\sigma$ of $\agent$
 satisfies the reach-avoid specification $(\discretestateset_r, \discretestateset_a)$ if $\discretestateinst_l \in \discretestateset_r$ and $\discretestateinst_{i,j} \notin \discretestateset_a$, for all $i \in [0, l]$ and $j \in [z_i]$.
	
	To be able to relate the behaviors of system~(\ref{eq:sys}) and that of a discrete one, we use the concept of simulation relations~\cite{TIOAmon}. 
 We define {\em forward simulation relations}
	as follows.
	
	\begin{definition}[Forward simulation relation]
	\label{def:fsr}
		Fix 
		(a) a discrete system $\agent = (\discretestateset, \disinputset, \delta)$;
		(b) an injective map $\inputmap: \disinputset \rightarrow \inputset$; (c) a time period $\timestep > 0$; 
		(d) a set of ordered pairs $\fsr \subseteq \stateset \times \discretestateset$.
		We say that $\fsr$ is a {\em forward simulation relation} (FSR) from the tuple $(\stateset,\inputset, \distset, f, \beta, \tau)$ to $\agent$,
		where $f$ is the right-hand-side of system~(\ref{eq:sys}),
		if 
	    (1) $\forall \stateinst \in \stateset$:
	    there exists $ \discretestateinst \in \discretestateset$, where $(\stateinst, \discretestateinst) \in \fsr$, and
	    (2) $\forall \stateinst \in \stateset$ and
	    $\discretestateinst \in \discretestateset$ such that $(\stateinst, \discretestateinst) \in \fsr$, 
     $\disinputinst_0 \in \disinputset$, 
     $t_e \in [0,\timestep]$, and 
	    a measurable $\distinst: [0,t_e] \rightarrow \distset$: 
	    there exists an execution 
	    $\sigma = \langle (\{ \discretestateinst_0\}, \discretestateinst_0, \disinputinst_0), (\{ \discretestateinst_{1,j}\}_{j\in [z_{1}]}, \discretestateinst_{1}, \bot) \rangle$ 
	    such that
	    $\forall t \in [0, t_e]$, $\exists j \in [z_1]$ such that $(\xi(\stateinst, \inputinst, \distinst; t),$ $s_{1,j}) \in \fsr$, where $\inputinst$ is the input signal 
     that is
     equal to $\beta(\disinputinst_0)$ over $[0,t_e]$,
	    and if $t_e = \timestep$, then $(\xi(\stateinst, \inputinst, \distinst; \timestep), s_1) \in \fsr$.
	\end{definition}

 Let $\constantinputsignals$ be the
	set  of left-piecewise-constant functions which only switch at time periods of $\timestep$ time units, mapping $\nnreals$ to the set $\inputmap(\disinputset)$. Then, the existence of a FSR implies that for every trajectory $\xi$ of system~(\ref{eq:sys}), with a control signal in 
	$\constantinputsignals$, there exists a corresponding execution $\sigma$ of the discrete system and $\sigma$ is said to {\em represent} $\xi$. If an FSR exists, we say that $\agent$ is an {\em abstraction} of $(\stateset,\inputset, \distset, f, \beta, \tau)$. 

	 Next, given reach and avoid sets $\stateset_r$ and $\stateset_a$ for system~(\ref{eq:sys}), we {\em abstract} them to the sets $\discretestateset_r$ and $\discretestateset_a$ of discrete states in $\discretestateset$. Such sets are not unique, but should satisfy:
  (1) $\forall \stateinst \in \stateset$:
	    if $\exists \discretestateinst \in \discretestateset_r$ such that $(\stateinst, \discretestateinst) \in \fsr$, then $\stateinst \in \stateset_r$, and
	(2) $\forall \stateinst \in \stateset_a$: 
	    the set $\{ \discretestateinst \in \discretestateset \ |\ (\stateinst, \discretestateinst) \in \fsr\}$ is non-empty and is a subset of $\discretestateset_a$.
 If an execution of $\agent$ satisfies the reach-avoid specification  $(\discretestateset_r, \discretestateset_a)$,  then all trajectories of system~(\ref{eq:sys}) represented by $\sigma$ satisfy the specification $(\stateset_r,\stateset_a)$. Further details are in Appendix~\ref{sec:specification_correspondence_discrete_continuous}.

	\vspace{-0.1in}
	\subsection{Discrete control synthesis}
	\label{sec:discrete_synthesis}
	
	In this section, we present a control synthesis algorithm for discrete abstractions of system~(\ref{eq:sys}).
	First, we define a controller of a discrete system $\agent$ to be a function  $\controller$ that maps $\discretestateset$ 
	to the set $\disinputset \cup \{\bot\}$. The {\em composition} of $\agent$ and $\controller$ results in a closed non-deterministic system, which we denote by the pair $(\agent,\controller)$. It can only start from a state in the set $R := \{ \discretestateinst \in \discretestateset | \controller[\discretestateinst] \neq \bot\}$.
	The executions of $(\agent,\controller)$ will have the same type as those of $\agent$ 
 with the actions taken following $\controller$, i.e., $\disinputinst_i = \controller[\discretestateinst_i]$.
	
	We say that $\controller$ satisfies a specification of $\agent$ if 
	$\forall \discretestateinst \in R$, all executions starting from $\discretestateinst$ of the system $(\agent, \controller)$ satisfy the specification.
	A trivial one 
	would map all $\discretestateinst \in \discretestateset$ to $\bot$, i.e., $R = \emptyset$. We seek non-trivial controllers with $R$ covering large parts of $\discretestateset$.
	
	If $\agent$ is an abstraction of  $(\stateset,\inputset, \distset, f, \beta, \tau)$
 with a FSR $\fsr$, then a corresponding controller $\controller$ {\em induces} a {\em zero-order-hold} controller $g$ for system~(\ref{eq:sys}). The controller $g$ periodically samples the state $\stateinst$ of system~(\ref{eq:sys})
	 every $\timestep$ time units and outputs 
	 $\beta(\controller[\discretestateinst])$ for the next period, where  $(\stateinst, \discretestateinst) \in \fsr$~\cite{tabuada_2009}.

	Given a discrete abstraction of system~(\ref{eq:sys}), one can use existing model checking tools to generate a corresponding controller $\controller$, such as the one   in Algorithm~\ref{code:discrete_synthesis} (see ~\cite{tabuada_2009}). 
	\begin{algorithm}
			\caption{Discrete control synthesis for reach-avoid specifications} 
			\label{code:discrete_synthesis}
			\begin{algorithmic}[1]
					\State \textbf{input:} $\agent = (\discretestateset, \disinputset, \delta), \discretestateset_a, \discretestateset_r$ \label{ln:input} 
				\State $\forall s \in \discretestateset, \controller[\discretestateinst] \gets \bot$;  $R \gets  \discretestateset_r$\label{ln:reachability_game_initialization}
				\While{$R$ did not reach a fixed-point} \label{ln:while_reachability_game}
				\State $R \gets \{ s \in \discretestateset | \exists \disinputinst \in \disinputset, \delta.\lastreach(\discretestateinst, \disinputinst) \subseteq R, \delta.\fullreach(\discretestateinst,\disinputinst) \cap \discretestateset_a= \emptyset\}$
				\EndWhile
				\State For each $\discretestateinst \in R$, assign $\controller[\discretestateinst]$ to its corresponding $\disinputinst \in \disinputset$.
				\State {\bf return: } $\controller$
			\end{algorithmic}
		\end{algorithm}
	
\vspace{-0.25in}
	\paragraph{Grid-based discrete abstractions}
	\label{sec:gridding}
 A common approach to generate a discrete abstraction of system~(\ref{eq:sys}) is to grid the compact state and control sets $\stateset$ and $\inputset$. 
 The resulting sets of states and control values $\discretestateset_\gridunder$ and $\disinputset_\gridunder$ have cardinalities equal to the number of cells in $Q_\state$ and $Q_\inputinst$, which we denote by $|Q_\state|$ and $|Q_\inputinst|$.
 The sets of reach and avoid states $\discretestateset_{\gridunder, r}$ and $\discretestateset_{\gridunder, a}$ consist of the states corresponding to the cells in $Q_\state$ that are subsets of $\stateset_r$ and intersect $\stateset_a$, respectively. 
 The FSR $\fsr_\gridunder$ maps any $\stateinst \in \stateset$ to the state in $\discretestateset_\gridunder$ corresponding to the cell in $Q_\state$ that $\stateinst$ belongs to.  The function $\inputmap_{\gridunder}$ maps each state in $\disinputset_\gridunder$ to the center of a unique cell in $Q_\inputinst$. Throughout the paper, we fix $\disinputset_\gridunder$ and $\inputmap_\gridunder$ and use them for all abstractions and drop their $\gridunder$ annotation. For any $(\discretestateinst_\gridunder, a) \in \discretestateset_\gridunder \times \disinputset$, one can compute $\delta_\gridunder$ for that pair by computing the reachable set of system~(\ref{eq:sys}) starting from the cell in $Q_\state$ corresponding to $\discretestateinst_\gridunder$ 
following $\beta(\disinputinst)$ for $\tau$ time units, using one of the existing reachable set-computation tools,  such as TIRA~\cite{TIRA_Meyer_2019}.
 \subsection{Symmetry-based discrete transitions computation}
	\label{sec:gridding}
	In Algorithm~\ref{code:symmetry_reachability_computations}, we present a more efficient way to compute the entries of $\delta_\gridunder$.
Algorithm~\ref{code:symmetry_reachability_computations} follows Corollary~\ref{cor:lower_dimensional_reachable_sets} and pre-computes the reachable sets starting from the cells 
 of the grid $\redgrid_{\state}$, which is $Q_{\state}$ 
	projected onto the cross-section $\symcrosssection$ defined by the moving frame $\gamma$.
 In other words, $\redgrid_{\state}$ is a grid over the set $\cup_{\stateinst \in \stateset} \phi_{\gamma(\stateinst)}(\stateinst)$ that has the same resolution as  $Q_\state$. The reachable sets are computed assuming disturbance sets of the form $\reddistset_j :=  \cup_{\stateinst_0 \in \stateset_j} \psi_{\gamma(\stateinst_0)}(\distset)$, where $\stateset_j = \cup_{i \in H_j}\ Q_{\state,i}$
and $H_j$ is the set of cells in $Q_\state$ that are projected to $\redgrid_{\state,j}$ in $\redgrid_\state$. It stores them in the data structure $\redreachset$. 
When a transition from an arbitrary $\discretestateinst_\gridunder \in \discretestateset_\gridunder$ \textbackslash$\discretestateset_a \cup \discretestateset_r$ and $\disinputinst \in \disinputset$ is needed during synthesis, Corollary~\ref{cor:lower_dimensional_reachable_sets} can be used to
  transform 
 the reachable set in $\redreachset$ that starts from the projection of $Q_{\state,i}$, the cell in $Q_\state$ corresponding to $\discretestateinst_\gridunder$, to $\symcrosssection$ from the relative coordinates back to the absolute ones. 
 Computing reachable sets  starting from lower-dimensional initial sets and transforming them is usually  computationally cheaper than computing ones starting from higher-dimensional initial sets. In addition, using this approach enables 
 the synthesis algorithm to compute exponentially-fewer reachable sets using reachability analysis engines, and transform the rest, potentially leading  to significant speedups.  
\vspace{-0.2in}

 	\begin{algorithm}
			\caption{Symmetry-based discrete transitions computation} 
			\label{code:symmetry_reachability_computations}
			\begin{algorithmic}[1]
			\For{$j \in  [|\redgrid_\state|] $ and $\disinputinst \in \disinputset$}
			 $\redreachset[j][a] \gets \reachset(\redgrid_{\state,j}, \beta(\disinputinst), \reddistset_j; [0,T])$ \label{ln:relative_reachset_computation}
			\EndFor
   \label{ln:relative_end_generate_transitions}
 			\end{algorithmic}
		\end{algorithm}

	\vspace{-0.4in}
	\section{Symmetry-based Abstraction}
	\label{sec:symmetry_based_abstraction_main}
	In this section, we describe a symmetry-based procedure that generates an abstraction $\agent_\abstractunder$ for $\agent_\gridunder$, shown in Algorithm~\ref{code:abstraction_quantized}.
 It takes as input a set $\discretestateset\undersym$ which has a one-to-one correspondence  with $\discretestateset_\gridunder$. We construct $\discretestateset\undersym$  in Algorithm~\ref{code:relative}.
	
	\paragraph{From absolute to relative coordinates}
\vspace{-0.2in}
  \begin{algorithm}
		\caption{Generation of (state, avoid set, reach set) tuples in relative coordinates } 
		\label{code:relative}
		\begin{algorithmic}[1]
		    \State {\bf input:} 
		    $Q_x$, $\stateset_r$, $\stateset_a$, $\{\phi_\alpha\}_{\alpha \in \liegroup}$, $\gamma$ 
		    \State $\discretestateset\undersym \gets \emptyset $
			\For{$i \in [|Q_\state|]$} \label{ln:relative_generate_Xsym}
			\State $\redgrid_{\state, i} \gets \cup_{\stateinst \in Q_{\state,i} }\phi_{\gamma(\stateinst)}(\stateinst)$; 
			 $\bar{\stateset}_a \gets \cup_{\stateinst \in Q_{\state,i}} \phi_{\gamma(\stateinst)}(\stateset_a)$; $\bar{\stateset}_r \gets \cap_{\stateinst \in Q_{\state,i}} \phi_{\gamma(\stateinst)}(\stateset_r)$
			\State $s\undersym \gets (\redgrid_{\state, i}, \bar{\stateset}_a, \bar{\stateset}_r)$; $\discretestateset\undersym \gets \discretestateset\undersym \cup \{s\undersym\}$ \label{ln:relative_last_line_generate_Xsym}
			\EndFor
			\State {\bf return: $\discretestateset\undersym$}
		\end{algorithmic}
	\end{algorithm}
	

	Similar to $\agent_\gridunder$, Algorithm~\ref{code:relative} creates a discrete state for each cell in $Q_\state$. For each $i \in [|Q_\state|]$, it creates a state $\discretestateinst\undersym \in \discretestateset\undersym$, which is a triple. The first entry, $\discretestateinst\undersym.\redgrid_{\state}$,
	is the $i^{\mathit{th}}$ 
     cell of $Q_{\state}$, $Q_{\state, i}$, 
projected onto 
 $\symcrosssection$  defined by $\gamma$.
	The set $\discretestateinst\undersym.\redgrid_{\state}$ is  $Q_{\state, i}$ 
	represented in the frames defined by its own states. Recall from Section~\ref{sec:cartan} that all states in $\discretestateinst\undersym.\redgrid_{\state}$ would have their first $r$ dimensions equal to a constant $c$.
	The second entry of $\discretestateinst\undersym$, $\discretestateinst\undersym.\redstateset_a$, represents the union of all the sets resulting from transforming $\stateset_a$ to the coordinate frames defined by the states in $Q_{\state, i}$. Thus, for each $\stateinst \in Q_{\state, i}$, $\discretestateinst\undersym.\redstateset_a$  is an over-approximation of $\stateset_a$ transformed by $\phi_{\gamma(\stateinst)}$. 
	The third element  $\discretestateinst\undersym.\redstateset_r$  represents the intersection of all the sets resulting from transforming $\stateset_r$ to the frames defined by the states in $Q_{\state, i}$. For each $\stateinst \in Q_{\state, i}$, $\discretestateinst\undersym.\redstateset_r$ is
	an under-approximation of $\stateset_r$ represented in the frame $\gamma(\stateinst)$, i.e., $\phi_{\gamma(\stateinst)}(\stateset_r)$.

 
	\paragraph{Generating the symmetry-based abstraction}

	\begin{algorithm}
		\caption{Symmetry-based abstraction construction} 
		\label{code:abstraction_quantized}
		\begin{algorithmic}[1]
		\State {\bf input: } $\discretestateset\undersym$, $\redreachset$, $\mathit{M} $ 
	\State $\discretestateset_\abstractunder \gets \{\discretestateinst_{\abstractunder,r}, \discretestateinst_{\abstractunder,a}\}$; 
 $ \forall \discretestateinst\undersym \in \discretestateset\undersym$, $\rv[\discretestateinst\undersym] \gets \bot$
			\For{$s\undersym \in \discretestateset\undersym$}
		    \If{$\discretestateinst\undersym.\redgrid_\state \subseteq \discretestateinst\undersym.\redstateset_r$}
      $\rv[\discretestateinst\undersym] \gets \discretestateinst_{\abstractunder,r}$, {\bf continue} \label{ln:symmetry_abstraction_reach_symbol_quantized}
		    \ElsIf{$\discretestateinst\undersym.\redgrid_\state \cap \discretestateinst\undersym.\redstateset_a \neq \emptyset$}
      $\rv[\discretestateinst\undersym] \gets \discretestateinst_{\abstractunder, a}$, {\bf continue} \label{ln:symmetry_abstraction_avoid_symbol_quantized}
		    \EndIf
            \State Let $j$ be the index of $\discretestateinst\undersym$ in $\redreachset$\label{ln:symmetry_abstraction_finding_index_in_reachdict}
            \State 
            $\mathit{obstructed} \gets \emptyset$ \label{ln:symmetry_abstraction_initializing_obstructed}
            \While{$|\mathit{obstructed}| < M$} \label{ln:while_over_obstructed}
              \State $a^* \gets \arg\!\min_{a \in \mathcal{U} \textbackslash \mathit{obstructed}} (\redreachset[j][a].\lastreach, \discretestateinst\undersym.\redstateset_r)$ \label{ln:symmetry_abstraction_finding_index_quantized_target}
              \If{$\discretestateinst\undersym.\redstateset_a \cap \redreachset[j][a].\fullreach \neq \emptyset$} \label{ln:symmetry_abstraction_checking_obstacle_intersection} 
              $\mathit{obstructed} \gets \mathit{obstructed} \cup \{a^*\}$
              \label{ln:symmetry_abstraction_adding_index_to_obstructed}
              \Else
              \If{$ (j,a^*) \notin \discretestateset_\abstractunder$}
              $\discretestateset_\abstractunder \gets \discretestateset_\abstractunder \cup \{ (j,a^*)\}$;   \label{ln:symmetry_abstraction_creating_new_abstract_state}
              \EndIf
              \State $\rv[\discretestateinst\undersym] \gets (j,a^*)$;
               {\bf break}
              \EndIf
            \EndWhile
            \If{$|\mathit{obstructed}| = M$} \label{ln:symmetry_abstraction_if_reached_obstructed_threshold}
             $\rv[\discretestateinst\undersym] \gets \discretestateinst_{\abstractunder, a}$ \label{ln:symmetry_abstraction_adding_to_avoid_symbol_quantized}
            \EndIf
		    \EndFor
		    \State {\bf return: } $\discretestateset_\abstractunder$, $\rv$ \label{ln:symmetry_abstraction_return}
		\end{algorithmic}
	\end{algorithm}

	Algorithm~\ref{code:abstraction_quantized} is an abstraction procedure that takes as input the 
 output $\discretestateset\undersym$ of Algorithm~\ref{code:relative}, 
 the dictionary $\redreachset$ computed in Algorithm~\ref{code:symmetry_reachability_computations}, and an integer $M \in [|\mathcal{U}|]$.
	It outputs a dictionary $\rv$ and a set $\discretestateset_\abstractunder$. $\rv$ maps every state in $\discretestateset\undersym$ (equivalently, $\discretestateset_\gridunder$) to some state in $\discretestateset_\abstractunder$. The subscript $\mathit{gs}$ standing for {\em grid}-based to {\em symmetry}-based abstraction.
 

	It  proceeds as follows. It first initializes $\discretestateset_\abstractunder$ to be a set of two states: $\discretestateinst_{\abstractunder,r}$, the {\em reach} state, and $\discretestateinst_{\abstractunder,a}$, the {\em avoid} state. Then, it updates $\rv$ in line~\ref{ln:symmetry_abstraction_reach_symbol_quantized} to represent any  $\discretestateinst\undersym \in \discretestateset\undersym$ whose 
 $\discretestateinst\undersym.\redstateset_r$ contains 
 $\discretestateinst\undersym.\redgrid_\state$ by
	$\discretestateinst_{\abstractunder,r}$ in $\discretestateset_\abstractunder$. Such states correspond to the cells in $Q_\state$ that 
 are inside 
 $\stateset_r$. Similarly, it updates $\rv$ in line~\ref{ln:symmetry_abstraction_avoid_symbol_quantized} to represent any  $\discretestateinst\undersym$ whose 
 $\discretestateinst\undersym.\redstateset_a$ intersects 
 $\discretestateinst\undersym.\redgrid_\state$, by $\discretestateinst_{\abstractunder,a}$. These are the cells in $Q_\state$ which intersect $\stateset_a$. It later maps more cells in $Q_x$ to $\discretestateinst_{\abstractunder,a}$ in line~\ref{ln:symmetry_abstraction_adding_to_avoid_symbol_quantized}. 
 In our implementation of Algorithm~\ref{code:abstraction_quantized}, we store the {\em last} reachable sets in $\redreachset[j][:]$ in an R-tree~\cite{rtree_original_paper}\footnote{We use \url{ https://github.com/libspatialindex/libspatialindex}.}.
Now, if some $\discretestateinst\undersym$ is not mapped to $\discretestateinst_{\abstractunder,r}$ or $\discretestateinst_{\abstractunder,a}$,  Algorithm~\ref{code:abstraction_quantized} 
either maps it to another abstract state in $\discretestateset_\abstractunder$ 
or creates a new one in lines~\ref{ln:symmetry_abstraction_finding_index_in_reachdict}-\ref{ln:symmetry_abstraction_adding_to_avoid_symbol_quantized}. 
 Algorithm~\ref{code:abstraction_quantized} first retrieves the index $j$ of the cell  $\discretestateinst\undersym.\redgrid_\state$ in the grid  $\redgrid_\state$ in line~\ref{ln:symmetry_abstraction_finding_index_in_reachdict}. 
 Then, in lines~\ref{ln:symmetry_abstraction_initializing_obstructed}-\ref{ln:symmetry_abstraction_creating_new_abstract_state}, it finds the control symbol that results in the closest reachable set in $\redreachset[j][:]$ to the relative target $\discretestateinst\undersym.\redstateset_r$ that does not intersect with relative obstacles $\discretestateinst\undersym.\redstateset_a$, according to some distance metric. Although this problem can be formulated as a quadratic program when using Euclidean distance, the number of constraints would be $O(|\disinputset|)$, which is exponential in the input dimension, and thus computationally expensive. Instead, in our implementation of  Algorithm~\ref{code:abstraction_quantized}, we approximate the minimizer. Our implementation first finds the closest point in the relative target, where the latter is a polytope, to the center of $\discretestateinst\undersym.\redgrid_x$
 by solving a simple quadratic program with the constraints being the half spaces of the polytope. 
 Then, it retrieves a set of closest hyperrectangles of a predetermined size in the R-tree storing the last reachable sets in $\redreachset[j][:]$ and their associated control symbols. This computation  is efficient as an R-tree is designed to optimize such a task~\cite{rtree_original_paper}. We maintain a set of {\em visited} control symbols that have been investigated in previous iterations of the while-loop, and remove them from the retrieved set. At every iteration, we increase the size of the set we request from the R-tree, as some of the returned ones would have been visited in previous iterations. 
 Algorithm~\ref{code:abstraction_quantized} also maintains a set $\mathit{obstructed}$ of control symbols  that have been found to result in reachable sets that intersect $\discretestateinst\undersym.\redstateset_a$.
 After that, it iteratively checks if there exists an $a^*$ among the retrieved control symbols results in a full reachable set that  does not intersect $\discretestateinst\undersym.\redstateset_a$.
 If that is the case, a new abstract state, denoted by  $(j,a^*)$, is created if it is not already in $\discretestateset_\abstractunder$. 
After that $\discretestateinst\undersym$ is mapped to that state in $\rv$, i.e., it sets $\rv[\discretestateinst\undersym]$ to $(j,a^*)$. Otherwise, 
they are all added to $\mathit{obstructed}$.
If  $|\mathit{obstructed}|$ becomes equal to $M$, the while-loop terminates and $\discretestateinst\undersym$ is mapped to  $\discretestateinst_{\abstractunder, a}$ in $\rv$.

	\section{Symmetry-based control synthesis algorithm}
	
	
	In this section, we present the main contribution of the paper: a control synthesis algorithm that exploits the symmetry-based abstraction as a guide to accelerate the search for specification-satisfying control symbols.
 The pseduocode is in Algorithm~\ref{code:ssynthesis_quantized}. 
	\begin{algorithm}

		\caption{Control synthesis using symmetry-based abstractions} 
		\label{code:ssynthesis_quantized}
		\begin{algorithmic}[1]
		    \State {\bf input:} $Q_\state, Q_\inputinst$, $\discretestateset\undersym, \redreachset$, $\discretestateset_\abstractunder, \rv$, $N$
        \State $\forall \discretestateinst_\abstractunder, \cache[\discretestateinst_\abstractunder] \gets  \{ (0, \discretestateinst_\abstractunder[1])\}$; $\forall \discretestateinst\undersym \in \discretestateset\undersym, \controller\undersym[\discretestateinst\undersym] \gets \bot$ \label{ln:ssynthesis_quantized_initialize_cache} 
     \State $R \gets \rv^{-1}[\discretestateinst_{\abstractunder,r}]$; 
     $E \gets \discretestateset\undersym$ \textbackslash\ $(R \cup \rv^{-1}[\discretestateinst_{\abstractunder,a}])$;
     $\mathit{progress} \gets 1$ 
		    \While{$\mathit{progress}$} \label{ln:ssynthesis_quantized_while}
            \State $\mathit{progress} \gets 0$ \label{ln:ssynthesis_quantized_initialize_progress_boolean}
            \For{ $\discretestateinst\undersym \in E$ } \label{ln:ssynthesis_quantized_for_loop_over_concrete_states}
            \State $A \gets$ control symbols in $\cache[\rv[\discretestateinst\undersym]]$ \label{ln:ssynthesis_quantized_initialize_A_to_cache}
            \For{$\mathit{itr} \in [N]$} \label{ln:ssynthesis_quantized_for_over_batches}
            \For{$a \in A$} \label{ln:ssynthesis_quantized_for_over_actions}
            \If{$\delta\undersym.\lastreach(\discretestateinst\undersym, \disinputinst) \subseteq R, \delta\undersym.\fullreach(\discretestateinst\undersym,\disinputinst) \cap \rv^{-1}[\discretestateinst_{\abstractunder,a}]= \emptyset$} \label{ln:ssynthesis_quantized_if_reach_avoid}
            \State $R \gets R \cup \{\discretestateinst\undersym\}$; 
            $\controller\undersym[\discretestateinst\undersym] \gets \disinputinst$; 
            $\mathit{progress}\gets 1$; 
            \If{$\exists\ \mathit{count} \geq 0$, $( \mathit{count},\disinputinst) \in \cache[\rv[\discretestateinst\undersym]]$}
            \State $\cache[\rv[\discretestateinst\undersym]].\mathit{replace}((\mathit{count},\disinputinst), (\mathit{count} + 1,\disinputinst))$ \label{ln:ssynthesis_quantized_cache_update}
            \Else\ 
             $ \cache[\rv[\discretestateinst\undersym]] \gets \cache[\rv[\discretestateinst\undersym]].\mathit{insert} ((1,\disinputinst))$  \label{ln:ssynthesis_quantized_cache_insertion} 
            \EndIf
            \State {\bf break} \label{ln:ssynthesis_quantized_endwhile}
            \EndIf
            \EndFor
            \If{$\discretestateinst\undersym \notin R$}
             $A \gets $ new (greedy or random) subset of $\disinputset $ \label{ln:ssynthesis_quantized_new_control_set_to_explore}
            \EndIf
            \EndFor 
		     \EndFor
            \State $E \gets$ neighborhood of newly added states to $R$ \label{ln:ssynthesis_quantized_updating_E}
            \EndWhile \label{ln:ssynthesis_quantized_endwhile}
			\State {\bf return } $\controller\undersym$
		\end{algorithmic}
	\end{algorithm}
	In addition to the outputs of Algorithms~\ref{code:symmetry_reachability_computations},~\ref{code:relative}, and~\ref{code:abstraction_quantized}, Algorithm~\ref{code:ssynthesis_quantized} takes 
 an upper bound $N$ on the number of different subsets of $\disinputset$ to be searched for a specification-satisfying control symbol for each state symbol in each synthesis iteration.
 It  outputs a controller $\controller\undersym$ for $\agent\undersym$ (equivalently, $\agent_\gridunder$) satisfying the specification $(\rv^{-1}[\discretestateinst_{\abstractunder,r}], \rv^{-1}[\discretestateinst_{\abstractunder,a}])$.

	Algorithm~\ref{code:ssynthesis_quantized} proceeds as follows: 
 it initializes the dictionary $\cache$ that maps each $\discretestateinst_\abstractunder \in \discretestateset_\abstractunder$ to a list of pairs of non-negative integers (scores) and control symbols in $\disinputset$. That list is maintained to be sorted in decreasing order of the scores. For each such pair, the first element, i.e., the score, represents the number of times its second element, i.e., the control symbol, has been found to be specification-satisfying for a state in $\rv^{-1}[\discretestateinst_\abstractunder]$. The list is initialized to have a single pair $(0, \discretestateinst_\abstractunder[1])$, where $\discretestateinst_\abstractunder[1]$ is the second element, the control, in the pair defining $\discretestateinst_\abstractunder$.
 Also, as Algorithm~\ref{code:discrete_synthesis}, it  initializes the set $R$ to be $\discretestateset_{\mathit{rel},r}$ and $\controller\undersym$ to a dictionary mapping every $\discretestateinst\undersym \in \discretestateset\undersym$ to $\bot$. In addition, Algorithm~\ref{code:ssynthesis_quantized} initializes a set $E$ to the set of  states in  $\discretestateset\undersym$ that are not reach or avoid ones.
	
Algorithm~\ref{code:ssynthesis_quantized} iteratively synthesizes $\controller\undersym$  in lines~\ref{ln:ssynthesis_quantized_while}-\ref{ln:ssynthesis_quantized_endwhile}. 
 Instead of searching $\disinputset$ in an {\em arbitrary} order for a specification-satisfying control 
 as Algorithm~\ref{code:discrete_synthesis},
 Algorithm~\ref{code:ssynthesis_quantized} searches $\disinputset$ according 
 to the order they appear in the cache. 
 Particularly, it explores the symbols in $\disinputset$ in batches. The batch at any iteration  is stored in the set $A$. It is initialized at line~\ref{ln:ssynthesis_quantized_initialize_A_to_cache} to be the list of control symbols appearing 
 in the cache entry corresponding to the representative state $\rv[\discretestateinst\undersym]$ of $\discretestateinst\undersym$ in $\discretestateset_\abstractunder$. That list is always sorted in the decreasing order of the scores. 
 Thus, Algorithm~\ref{code:ssynthesis_quantized} checks the control symbols with higher scores first in the for-loop in line~\ref{ln:ssynthesis_quantized_for_over_actions}. If none of the symbols in the cache are specification-satisfying, Algorithm~\ref{code:ssynthesis_quantized} samples a new subset of control symbols in $\disinputset$ to investigate in the next iteration of the for-loop at line~\ref{ln:ssynthesis_quantized_for_over_batches}. It either samples that batch greedily or uniformly at random. The greedy approach is the same one we followed during the abstraction  process in Algorithm~\ref{code:abstraction_quantized}, i.e., the ones closest to the relative target.

 If some $a \in A$ satisfies the specification for a state $\discretestateinst\undersym$, then  Algorithm~\ref{code:ssynthesis_quantized} increments
 its score in $\cache[\rv[\discretestateinst\undersym]]$, if it exists, or adds a new pair $(1,a)$ to that entry of the cache,  otherwise. In either case, when $\cache[\rv[\discretestateinst\undersym]]$ is modified, it is re-sorted to remain  ordered in a decreasing order of scores.

 In an ideal scenario, for any $\discretestateinst_\abstractunder \in \discretestateset_\abstractunder$, the greedy control $\discretestateinst_\abstractunder[1]$ found in the abstraction process would suffice for satisfying the specification for all $\discretestateinst\undersym \in \rv^{-1}[\discretestateinst_\abstractunder]$. In that case, when Algorithm~\ref{code:ssynthesis_quantized} terminates, $\cache[\discretestateinst_\abstractunder]$ would be a singleton list containing the pair $(\mathit{count}, a)$, where $\mathit{count} = |\rv^{-1}[\discretestateinst_\abstractunder]|$ and $a = \discretestateinst_\abstractunder[1]$. 
 However, that might not be the case in most scenarios. Different states in $\rv^{-1}[\discretestateinst_\abstractunder]$  might require different control to achieve the specification.
 The control symbol $\discretestateinst_\abstractunder[1]$ is a greedy choice to minimize the distance to the target. However, as well known from optimal control theory and dynamic programming, a greedy decision is not necessarily the optimal one. There are non-symmetric state coordinates 
 and further away obstacles that affect the optimality of the control. Algorithm~\ref{code:ssynthesis_quantized} adds states that are closer to the target to the extended target first, same as Algorithm~\ref{code:discrete_synthesis}.  
 Thus, at a certain iteration of the while-loop in line~\ref{ln:ssynthesis_quantized_while} and iteration of the for-loop in line~\ref{ln:ssynthesis_quantized_for_loop_over_concrete_states}, the 
 greedily-chosen control symbol
 might not satisfy the specification while the optimal one does, which requires the for-loop in line~\ref{ln:ssynthesis_quantized_for_over_batches} to iterate until finding it. 
 Our intuition is that in most environments,
 states with the same symmetry-based abstract state will require a small set of control symbols to achieve the reach-avoid specification, resulting in cache entries with short lists after the termination of Algorithm~\ref{code:ssynthesis_quantized}. That allows considering large $\disinputset$ without incurring the expense of exploring it fully in an arbitrary order for every state in $\discretestateset\undersym$ at every synthesis iteration. 

 Additionally, after each synthesis iteration, a non-optimized algorithm such as Algorithm~\ref{code:discrete_synthesis} would only remove the states that have been added to $R$ in that iteration from $E$, i.e., update $E$ to be $E$ \textbackslash\ $R$. However, given all the reachable sets one would need to compute $\delta_\gridunder$,
  they can compute the maximum travelled distance in each dimension in one step starting from any state and following any control symbol.
  Notice that it would sufficient for Algorithm~\ref{code:ssynthesis_quantized} to update $E$ in line~\ref{ln:ssynthesis_quantized_updating_E} to be the states in $\discretestateset\undersym$ that correspond to cells in $Q_x$ that are within these maximum per-dimension distances from the newly added states to $R$ in the same synthesis iteration. That is because other states further away cannot be added to $R$ in the next synthesis iteration as their reachable sets following any control cannot intersect the newly added states to $R$. 
 This significantly reduces the number of states that Algorithm~\ref{code:ssynthesis_quantized} has to explore. 
 Now, assume that for any $\alpha \in \liegroup$, $\phi_\alpha$ is a rigid motion transformation, i.e., preserves the distance and angle measures. Then, the maximum travelled distance in each dimension
 in the reachable sets stored in $\redreachset$ would be the same as that in all the reachable sets needed if one had to construct $\delta_\gridunder$, and thus there is no need to compute the latter. That is what we follow in our implementation of Algorithm~\ref{code:ssynthesis_quantized}. 
Finally, consider the case when $N$ and the subsets selected in line~\ref{ln:ssynthesis_quantized_new_control_set_to_explore} are sufficiently large to cover all symbols in $\disinputset$ in every iteration of the for-loop over $E$ 
in line~\ref{ln:ssynthesis_quantized_for_loop_over_concrete_states}. Also, assume that computing the reachable sets for the transitions by transforming those in $\redreachset$ using Corollary~\ref{cor:lower_dimensional_reachable_sets} results in the same over-approximation errors as computing them from scratch using existing tools. Then, Algorithm~\ref{code:ssynthesis_quantized} would be the same as Algorithm~\ref{code:discrete_synthesis} with the only two differences being the order by which the control symbols are explored and the restriction of the states to explore at a given iteration to those whose neighborhood has changed in the past iteration. Thus, it retains the same correctness guarantees as Algorithm~\ref{code:discrete_synthesis}.
	\section{Case study}
	\label{sec:experiments}
	In this section, we describe our results from synthesizing a controller for a 1:30 scale model of a platform supply vessel using Algorithm~\ref{code:ssynthesis_quantized}. 
 The dynamics of this model have been described in~\cite{marine_craft_handbook} and used as a case study in~\cite{Meyer2020_control_synthesis_ship_docking}.
	Its kinematics can be approximated using the ODE~\cite{marine_craft_handbook,Meyer2020_control_synthesis_ship_docking}: $\dot{\eta} = R(\theta) \nu + v_c$, 
		where $\eta = [N;E;\theta] \in \reals^3$ is the state of the ship representing its South-North and West-East position and its heading angle, respectively. The control input $\nu$ is 3-dimensional.
  $R(\theta)$ is a rotation matrix that given $\theta$, rotates $\nu$ from the body coordinates to the global ones, where $\theta = 0$ points to the North. Finally, $v_c$ represents the disturbance resulting from water-current velocities.
		We consider the same environment and reach-avoid specification of Meyer et al.~\cite{Meyer2020_control_synthesis_ship_docking}, where the ship is required to dock while avoiding obstacles in a rectangular pier. We make a slight modification to represent boundaries as obstacles. Further details are in Appendix~\ref{sec:appendix_numerical_details}.
		
		Let $\liegroup$ be $SE(2)$, the special Euclidean group, representing 2-dimensional rotations and 3-dimensional translations. Then, the kinematics of the ship 
  are $\liegroup$-equivariant, where 
		the transformation group $\{h_\alpha = (\phi_\alpha, \chi_\alpha, \psi_\alpha)\}_{\alpha\in \liegroup}$ is defined as follows:
		$ \phi_{\alpha}(\eta) = R_\alpha^\intercal (\eta - \eta_\alpha)$,
			$\chi_{\alpha}(\nu) = I \nu = \nu$,
			$\psi_{\alpha}(v_c) = R_\alpha^\intercal v_c$,
  where $\eta_\alpha$ and $R_\alpha$ are the 3-dimensional translation vector and rotation matrix corresponding to $\alpha$.
		We define the moving frame $\gamma$ to be the one that maps any $\stateinst = [N;E;\theta]$ to the $\alpha \in \liegroup$ with $R_\alpha = R(\theta)$ and $\eta_\alpha = \stateinst$. Thus, $\symcrosssection$ consists of a singleton state: the origin. 

  We implemented our algorithms 
  in Python. We used TIRA~\cite{TIRA_Meyer_2019} (implemented in Matlab) to compute the reachable sets in Algorithm~\ref{code:symmetry_reachability_computations}. The reachable sets are represented as lists of time-annotated axis-aligned hyper-rectangles in $\reals^3$. 
  We used the Polytope library\footnote{\url{https://pypi.org/project/polytope/}} to translate and rotate them. 
  %
  We ran Algorithm~\ref{code:abstraction_quantized} in parallel to generate the symmetry-based abstractions. It is a naturally parallelizable algorithm since the concrete states determine their representative abstract states independent from each other.
  We ran the control synthesis part of Algorithm~\ref{code:ssynthesis_quantized} sequentially, although it can be ran in parallel as well.  We considered several choices of  resolutions of the grid $Q_x$ and fixed $|Q_u| = 9^3$ in all experiments.  
  The results are shown in Table~\ref{tab:results_quantized}. 
\begin{table}[!htp]
    \centering
%
\caption{Experimental results.  The table presents results from runs of Algorithm~\ref{code:ssynthesis_quantized} with different parameters and resolutions of $Q_\state$. It shows: \# cells in $Q_x$ that are not avoid or reach ones $|\discretestateset_{\gridunder}^*| = |\discretestateset_{\gridunder} \textbackslash \discretestateset_{\gridunder,a} \cup \discretestateset_{\gridunder,r}|$; \# abstract states $|\discretestateset_{\abstractunder}|$; \# states mapped to $\discretestateinst_{\abstractunder,a}$ in line~\ref{ln:symmetry_abstraction_if_reached_obstructed_threshold} of Algorithm~\ref{code:abstraction_quantized} (\#CAO);  total \# states added to $R$ during synthesis (\#Ctr); min / average / median / max lengths of lists in $cache$; average of $|E|/(|\discretestateset_\gridunder^*| - |R| + |\discretestateset_{\gridunder,r}|)$ over synthesis iterations (Exp, for states-to-be-explored); average / max of the lengths of paths from states in $R$ to the reach set (|Path|);  the abstraction, synthesis, and total computation times in seconds (At, St, and Tt).\protect\label{tab:results_quantized}}
    \begin{tabular}{lr|rrrrrrrrrr}\toprule
         Str & $Q_x$ & $|\discretestateset_
        \gridunder^*|$ & $|\discretestateset_\abstractunder|$ & \#CAO & \#Ctr & $\cache$ & Exp &  |Path| & At & St & Tt \\  \midrule
        0 & $1$ & 6580 & 0 & 0 & 1750 & - & 1 & 4.6 / 11 & - & 2496  & 2496  \\
        0.5 & $1$ & 6580 & 0 & 0 & 1750 & - & 1 & 5 / 13 & - & 2437  & 2438  \\
        1 & 1 & 6580 & 143 & 210 & 1750 & 1 / 1.5 / 1.0 / 13 & 1 & 4.6 / 11 & 108 & 3683  & 3791  \\ 
        2 & 1 & 6580 & 143 & 210 & 1750 & 1 / 1.5 / 1.0 / 12 & 1 & 4.6 / 11 & 105 & 2209  & 2315  \\ 
        3 & 1 & 6580 & 35 & 210 & 1750 & 1 / 2.9 / 2.0 / 17 & 1 & 4.6 / 11 & 94 & 2590  & 2685  \\ 
        4 & 1 & 6580 & 143 & 210 & 1750 & 1 / 1.5 / 1.0 / 13 & 0.40 & 4.6 / 11 & 106 & 3017  & 3123  \\ 
        5 & 1 & 6580 & 143 & 210 & 1750 & 1 / 1.5 / 1.0 / 16 & 0.40 & 4.7 / 11 & 105 & 1373  & {\bf 1479}  \\ 
        6 & 1 & 6580 & 35 & 210 & 1750 & 1 / 2.9 / 2.0 / 17 & 0.40 & 4.6 / 11 & 97 & 2465  & 2562  \\ \hline
        0 & 2 & 34544 & 0 & 0 & 31056 & - & 1 & 10.2 / 28 & - & 11895 & 11895 \\ 
        0.5 & 2 & 34544 & 0 & 0 & 31056 & - & 1 & 10.5 / 30 & - & 11363 & 11393 \\
        1 & 2 & 34544 & 167 & 3488 & 31056 & 2 / 7.4 / 5.0 / 54 & 1 & 10.2 / 28 & 567 & 17096 & 17663 \\ 
        2 & 2 & 34544 & 167 & 3488 & 31056 & 1 / 7.3 / 4.5 / 47 & 1 & 10.4 / 30 & 515 & 9663 & 10179 \\ 
        3 & 2 & 34544 & 34 & 3488 & 31056 & 2 / 12.3 / 4 / 106 & 1 & 10.2 / 28 & 510 & 12072 & 12583 \\ 
        4 & 2 & 34544 & 167 & 3488 & 31056 & 2 / 6.7 / 5.0 / 37 & 0.32 & 10.2 / 29 & 581 & 13953 & 14534 \\ 
        5 & 2 & 34544 & 167 & 3488 & 31056 & 2 / 7.1 / 5.0 / 43 & 0.29 & 10.3 / 34 & 514 & 7335 & {\bf 7850} \\ 
        6 & 2 & 34544 & 34 & 3488 & 31056 & 2 / 11.6 / 4 / 97 & 0.32 & 10.2 / 29 & 517 & 11679 & 12197 \\ \hline
        0 & 3 & 58030 & 0 & 0 & 57180 & - & 1 & 11.0 / 32 & - & 20910 & 20938 \\
        5 & 3 & 58030 & 215 & 850 & 57180 & 1 / 7.1 / 4.0 / 50 & .47 & 10.3 / 32 & 831 & 12999 & {\bf 13832} \\  \hline
        0 & 4 & 117340 & 0 & 0 & 105500 & - & 1 & 8.6 / 24 & - & 39204 & 39260 \\ 
        5 & 4 & 117340 & 149 & 11840 & 105500 & 1 / 12.5 / 8.0 / 67 & 0.28 & 8.4 / 25 & 1870 & 21117 & {\bf 22989} \\ 
        \bottomrule
    \end{tabular}
\end{table}

We ran Algorithm~\ref{code:ssynthesis_quantized} using different parameters to identify the benefits of the different optimizations that we  introduced. We call each choice a {\em strategy} and try 8 of them. Strat. $0$ is the {\em baseline} and the traditional implementation of Algorithm~\ref{code:discrete_synthesis}  where states and controls are explored in an arbitrary order. Strat. $0.5$ differs from Strat. $0$ by exploring only 400 unique controls sampled uniformly at random i.i.d at every state in every synthesis iteration. The rest correspond to running Algorithm~\ref{code:ssynthesis_quantized} with $N$ and batch sizes chosen to explore all symbols (strategies $1$, $3$, $4$, and $6$) or 400 unique symbols (including those in the cache) (strategies $2$ and $5$) in $\disinputset$ in the loop in  line~\ref{ln:ssynthesis_quantized_for_over_batches}. We use the same number of allowed controls to explore in creating the abstractions in Algorithm~\ref{code:abstraction_quantized}.    Strategies $1$, $2$, and $3$ update $E$ in line~\ref{ln:ssynthesis_quantized_updating_E} to $E$ \textbackslash\ $R$ while $4$, $5$, and $6$ update it to the neighborhood of newly added states to $R$. Strategies $3$ and $6$ consider an arbitrary order of control symbols, instead of the greedy ones, for the abstraction and for control synthesis. In other words, for each $\discretestateinst\undersym \in \discretestateset\undersym$, Algorithm~\ref{code:abstraction_quantized} following that strategy iterates over the control symbols according to their indices and the first symbol that does not result in collision is used for the abstraction. Similarly, Algorithm~\ref{code:ssynthesis_quantized} following that strategy iterates over the controls in the cache first, and if none is specification-satisfying, it iterates over the rest in an arbitrary order.
%
Moreover, we considered different grid resolutions for $Q_x$: $1$ corresponds to $30\times 30\times 30$ ($30$ partitions per dimension); $2$ corresponds to $50 \times 50 \times 50$; $3$ corresponds to $60\times 70 \times 50$; and $4$ corresponds to $80 \times 100\times 50$. 

From Table~\ref{tab:results_quantized}, we can see that Strat. 5, the one that uses all optimizations we introduced (symmetry-abstraction for choosing and updating the cache, updating $E$ to neighborhoods of newly added states to $R$, and considering only subset of $\disinputset$ in the order of the scores) result in the least Tt with respect to all other strategies at almost no cost in \#Ctr and |Path|. It results in a 1.7$\times$, 1.5$\times$, 1.5$\times$, and 1.7$\times$ speedup over the benchmark runs for resolutions $1$, $2$, $3$, and $4$, respectively. 

Also, we can see that resolutions $2$, $3$, and $4$ were sufficient to find a specification-satisfying control for all states in $\discretestateset\undersym$ corresponding to states in $\discretestateset_\gridunder^*$ and the abstraction was able to find all states that cannot be controlled to satisfy the specification and map them to $\discretestateinst_{\abstractunder,a}$, before synthesis. Additionally, despite Strat. 6 (which explores control symbols in an arbitrary order for abstraction and synthesis) resulting in fewer abstract states than Strat. 5, the $\cache$ did not help it in decreasing St and resulted in a St similar to that of the benchmark.  If we add the overhead of At, it required more computation time than the baseline. Its sub-optimality is also evident from the lengths of the lists in the cache being much longer than those of Strat. 5, which demonstrate that more states represented by the same abstract state required different controls to satisfy the specification. Moreover, Strat. 5 is significantly faster than Strat. 2, which shows the benefit of restricting $E$ to the neighborhood of newly added states to $R$ instead of exploring all remaining states in each synthesis iteration. Yet, we found that iterating over the newly added states to $R$ after each synthesis iteration and taking the union of their neighborhoods, itself adds an overhead of computation time. Despite that overhead, it still resulted in more efficient computation time than without that optimization when combined with the symmetry-based abstraction and synthesis. Otherwise, the overhead seems to balance out its benefit (Strat. 6 versus Strat. 3). We also found that when exploring the control symbols during synthesis in a similar manner to how we construct the symmetry-based abstraction also adds an overhead for finding the closest reachable sets and for excluding those control symbols returned by the R-tree that are already explored. That overhead is worth it when exploring a subset of $\disinputset$, i.e., small $N$ and batch sizes (results of Strat. 2 and Strat. 5 versus Strat. 0.5). However, it is not when we are exploring all symbols in $\disinputset$ (results of Strat. 1 and Strat. 4 versus the baseline).
	\section{CONCLUSIONS}
	We proposed an algorithm to accelerate symbolic control synthesis for continuous-time dynamical systems and reach-avoid specifications. Our algorithm exploits dynamical symmetries even in non-symmetric environments through an abstraction approach. Our algorithm uses the abstraction to explore control symbols more selectively during synthesis.
 Our approach is general; it allows exploitation of any symmetries in the form of Lie groups that the system possesses. 
 We show experimental results demonstrating the effectiveness of our approach.

 \section{Acknowledgements}

 This work has been supported by the National Science Foundation under grant number CNS-2111688. 




	\bibliographystyle{abbrv}
	\bibliography{hussein}
	
	
 \appendix



  \section{Proof of Corollary~\ref{cor:lower_dimensional_reachable_sets}}
 \label{sec:appendix_proofs}

\symmetryreachablesets*

	\begin{proof}
	The proof follows from applying Theorem~\ref{thm:symmetric_reachsets} to the reachable sets starting from each initial state in $\stateset_0$: $\reachset(\stateset_0, \inputinst, \distset;[t_0,t_1]) = $
	    \begin{align}
		&\cup_{\stateinst_0 \in \stateset_0} \reachset(\stateinst_0, \inputinst, \distset;[t_0,t_1])  \nonumber \\
		&\hspace{0.5in}\text{[by the definition of $\reachset$]} \nonumber\\
		&= \cup_{\stateinst_0 \in \stateset_0} {\phi_{\gamma(\stateinst_0)}}^{-1}(\reachset(\phi_{\gamma(\stateinst_0)}(\stateinst_0), \inputinst, \psi_{\gamma(\stateinst_0)}(\distset);[t_0,t_1]))  \nonumber \\
		&\hspace{0.5in}\text{[by the definition of $\reachset$ and by Theorem~\ref{thm:symmetric_reachsets}]} 
  \nonumber \\		 
  &= \cup_{\stateinst_0 \in \stateset_0}{\phi_{\gamma(\stateinst_0)}}^{-1}( \reachset( \redstateinst_0,\inputinst, \psi_{\gamma(\stateinst_0)}(\distset);[t_0,t_1]))  \nonumber \\
		 &\hspace{0.5in}\text{[where $\redstateinst_0 \in \symcrosssection$]}  \nonumber \\
		&\subseteq \cup_{\stateinst_0 \in \stateset_0}{\phi_{\gamma(\stateinst_0)}}^{-1}(\reachset(\redstateset_0,\inputinst, \reddistset;[t_0,t_1])),
		 \end{align}
		where the last step follows from the fact that $\reachset(\redstateinst_0,\inputinst, \psi_{\gamma(\stateinst_0)}(\distset)) \subseteq$ $\reachset($ $\redstateset_0,\inputinst, \reddistset;[t_0,t_1])$ since $\redstateinst_0 \in \redstateset_0$ and $\psi_{\gamma(\stateinst_0)}(\distset) \subseteq  \reddistset$.
	\end{proof}

 \section{Specification correspondence between a continuous-time system and a corresponding discrete abstraction}
 \label{sec:specification_correspondence_discrete_continuous}

 Theorem~\ref{thm:fsr_executions_correspondance} below is an adaptation of Lemma 4.23 of \cite{TIOAmon} to our continuous-to-discrete abstraction setting. Theorem~\ref{thm:fsr_executions_correspondance} shows that the existence of a FSR implies that for every trajectory $\xi$ of system~(\ref{eq:sys}), there exists a corresponding execution $\sigma$ of the discrete system. In that case, we say that $\agent$ is an {\em abstraction} of $(\stateset,\inputset, \distset, f, \beta, \tau)$,
	with control signals constrained to be in 
	$\constantinputsignals$, the
	set  of left-piecewise-constant functions, which only switch at time periods of $\timestep$ time units, mapping $\nnreals$ to the set $\inputmap(\disinputset)$. 
	\begin{theorem}[Executions correspondence]
	\label{thm:fsr_executions_correspondance}
	For any
	$\stateinst \in \stateset, \inputinst \in \constantinputsignals,$ measurable $\distinst: \nnreals \rightarrow \distset$, and time horizon $T \geq 0$:
	there exists an execution $\sigma = \langle (\{ \discretestateinst_0\}, \discretestateinst_0, \disinputinst_0), \dots,$ $(\{ \discretestateinst_{l,j}\}_{j\in z_{l}}, \discretestateinst_{l}, \bot) \rangle$ of $\agent$, where 
	$ k = \lceil \frac{T}{\timestep}\rceil$,
	such that for all indices $i \in [l]$,
	$\beta(\disinputinst_i) = \inputinst(i\tau)$,  $(\xi(\stateinst, \inputinst, \distinst; i\timestep), s_i) \in \fsr$, and
	 for all indices $i \in [l]$,   $\forall t \in [i\tau, \min\{(i+1)\tau, T\}]$, $\exists j \in [z_{i+1}]$ such that $(\xi(\stateinst, \inputinst, \distinst; t), s_{i+1,j}) \in \fsr$.
	 \end{theorem}
  
	 We say that $\sigma$ in Theorem~\ref{thm:fsr_executions_correspondance} {\em represents} $\xi$. Also, we abuse notation and say that $\sigma \in \fsr(\xi)$ and $\xi \in \fsr^{-1}(\sigma)$. The proof of the theorem is similar to that of Lemma 4.23 of \cite{TIOAmon}, and skip it here for conciseness.

In the following corollary of  Theorem~\ref{thm:fsr_executions_correspondance}, we show that if a trajectory $\xi$ of system~(\ref{eq:sys}) violates the reach-avoid specification $(\stateset_r, \stateset_a)$, then any corresponding execution $\sigma$ of $\agent$, as defined in Theorem~\ref{thm:fsr_executions_correspondance}, violates the reach-avoid specification $(\discretestateset_r, \discretestateset_a)$ as well. Moreover, if an execution $\sigma$ satisfies the reach-avoid specification $(\discretestateset_r, \discretestateset_a)$, then any trajectory $\xi$ of system~(\ref{eq:sys}) that is related to $\sigma$ in $\fsr$ satisfies the  $(\stateset_r, \stateset_a)$ specification. 

\begin{corollary}[Specification correspondence]
\label{cor:specification_correspondence}
For any $\stateinst \in \stateset, \inputinst \in \constantinputsignals,$ measurable $\distinst: \nnreals \rightarrow \distset$, and time horizon $T \geq 0$:
if $\exists t \in [0, T]$ such that $\xi(\stateinst, \inputinst, \distinst; t) \in \stateset_a$, 
 let $i = \lfloor \frac{t}{\tau}\rfloor$, then
$\exists j \in [z_i]$, such that $\discretestateinst_{i,j} \in \discretestateset_a$. Moreover, if $\nexists t\in [0,T]$ such that $\xi(\stateinst, \inputinst, \distinst; t) \in \stateset_r$, then for all $\sigma \in \fsr(\xi)$, $\nexists  i \in [\lceil\frac{T}{\timestep}\rceil+1]$ and $j \in [z_i]$ such that $s_{i,j} \in \discretestateset_r$. 
On the other hand, fix any execution $\sigma$ of $\agent$ of length $l \geq 1$. Then, 
if $\nexists i \in [l+1]$ and $j \in [z_i]$, such that $\discretestateinst_{i,j} \in \discretestateset_a$, then for all $\xi \in \fsr^{-1}(\sigma)$,
$\nexists t \in [0, k\timestep]$ such that $\xi(\stateinst, \inputinst, \distinst; t) \in \stateset_a$. Moreover, if
$s_{l} \in \discretestateset_r$, then 
for any $\xi \in \fsr^{-1}(\sigma)$ of duration $l \timestep$,
$\xi(\stateinst, \inputinst, \distinst; l\timestep) \in \stateset_r$.
	\end{corollary}
	
	In addition to Theorem~\ref{thm:fsr_executions_correspondance}, the proof of Corollary~\ref{cor:specification_correspondence} follows from
	the definition of $\discretestateset_r$ which under-approximates $\stateset_r$ and the definition of $\discretestateset_a$ which over-approximates $\stateset_a$.


 The following corollary of Corollary~\ref{cor:specification_correspondence} shows that if the controller $\controller$ of $\agent$ satisfies a reach-avoid specification that abstracts the reach-avoid specification of system~(\ref{eq:sys}), then the controller it induces satisfies that specification of system~(\ref{eq:sys}).
	
	\begin{corollary}
	If there exists a discrete abstraction $\agent$ 
	of $(\stateset,\inputset, \distset, f, \beta, \tau)$, an abstraction $(\discretestateset_r, \discretestateset_a)$ of $(\stateset_r, \stateset_a)$, 
	and a controller $\controller$ that satisfies the discrete reach-avoid specification $(\discretestateset_r, \discretestateset_a)$ for $\agent$ starting from a set $R \subseteq \discretestateset$, then the zero-order-hold controller induced by $\controller$ satisfies the reach-avoid specification $(\stateset_r, \stateset_a)$ for system~(\ref{eq:sys}) starting from the set $\fsr^{-1}(R) \subseteq \stateset$.
	\end{corollary}



 \section{Details of the ship-docking case study}
 \label{sec:appendix_numerical_details}
 The kinematics of a marine vessel can be approximated using the ODE~\cite{Meyer2020_control_synthesis_ship_docking}: $\dot{\eta} = R(\theta) \nu + v_c$, 
		where $\eta = [N;E;\theta] \in \reals^3$ is the state of the ship representing its South-North and West-East position and its heading angle, respectively. The control input $\nu$ is a 3-dimensional vector that determines the surge, sway, and yaw velocities of the ship. $R(\theta)$ is a rotation matrix that given the heading angle $\theta$, rotates $\nu$ from body coordinates to the global coordinates where $\theta = 0$ points to the North. Finally, $v_c$ represents the disturbance resulting from water-current velocities.
		We consider the same environment and reach-avoid specification of Meyer et al.~\cite{Meyer2020_control_synthesis_ship_docking}, where the ship is required to dock while avoiding obstacles in a rectangular pier. We make a slight modification to represent boundaries as obstacles. 
	The original model of the ship is 6-dimensional. However, as shown in \cite{Meyer2020_control_synthesis_ship_docking}, three dimensions can be abstracted using another form of abstraction,
called {\em continuous abstraction}. 

 We implemented our algorithms 
  in Python. We used TIRA~\cite{TIRA_Meyer_2019} (implemented in Matlab) to compute the reachable sets in Algorithm~\ref{code:symmetry_reachability_computations}. The reachable sets are represented as lists of time-annotated axis-aligned hyper-rectangles in $\reals^3$. Then, we transformed their hyper-rectangles to polytopes before translating and rotating them to obtain the reachable sets starting from the different cells in $Q_\state$ to get the relevant transitions in $\delta_\gridunder$ during synthesis. We used the Polytope library\footnote{\url{https://pypi.org/project/polytope/}} to perform such transformations. 
  We over-approximate the union in Corollary~\ref{cor:lower_dimensional_reachable_sets} by transforming the reachable set in $\redreachset$ using the frames corresponding to the vertices of the hyper-rectangle representing $\stateset_0$, and then taking the bounding box of the resulting reachable sets.

     We follow the experimental setup of \cite{Meyer2020_control_synthesis_ship_docking}: we consider the rectangular state set $\stateset = [0,10] \times [0,6.5] \times [-\pi,\pi]$ and input set $\inputset = [-0.18,0.18] \times [-0.05,0.05] \times [-0.1,0.1]$. We use a disturbance set $\distset = [\frac{-0.01}{\sqrt{2}},\frac{0.01}{\sqrt{2}}]^3$ smaller than the one in \cite{Meyer2020_control_synthesis_ship_docking}, which was $[-0.01,0.01]^3$, since after the rotations in Corollary~\ref{cor:lower_dimensional_reachable_sets}, our $\distset$ will become the set $\reddistset = [-0.01,0.01]^3$. For the reach-avoid specification, as \cite{Meyer2020_control_synthesis_ship_docking}, we consider the reach set $\stateset_r = [7,10] \times [0, 6.5] \times [\pi/3, 2\pi/3]$ and the avoid sets $\stateset_a = \stateset_{a,0} \cup \stateset_{a,1}$, where $\stateset_{a,0}  = [2, 2.5] \times [0,3] \times [-\pi,\pi]$ and $\stateset_{a,1} = [5,5.5] \times [3.5, 6.5] \times [-\pi, \pi]$. 
		However, in contrast with \cite{Meyer2020_control_synthesis_ship_docking}, we  add more avoid sets to represent the boundaries of $\stateset$. The reason behind this change is to keep track of the environment boundaries in the symmetry-based abstraction.  
We do that as follows: first, we expand the boundaries of $\stateset$ by 3 units in each of the first two dimensions, resulting in $\stateset = [-3,13] \times [-3,9.5] \times [-\pi,\pi]$. Then, we expand the two avoid sets and add four new obstacles at each side of $\stateset$ resulting in $\stateset_a = \cup_{i \in [6]} \stateset_{a,i}$, where $\stateset_{a,0} = [2, 2.5] \times [-3,3] \times [-\pi,\pi]$, $\stateset_{a,1} = [5,5.5] \times [3.5, 9.5] \times [-\pi, \pi]$, $\stateset_{a,1} = [5,5.5] \times [3.5, 9.5] \times [-\pi, \pi]$, $\stateset_{a,2} = [-3,0] \times [-3, 9.5] \times [-\pi, \pi]$, $\stateset_{a,3} = [-3,13] \times [-3, 0] \times [-\pi, \pi]$, $\stateset_{a,4} = [-3,13] \times [-3, 9.5] \times [-\pi, \pi]$, and $\stateset_{a,5} = [10,13] \times [6.5, 9.5] \times [-\pi, \pi]$. Since our environment is larger, having the same number of partitions per dimensions in $Q_x$ as Meyer et al.~\cite{Meyer2020_control_synthesis_ship_docking} results in larger cells in our case. Hence, we consider a finer gridding ($80 \times 100 \times 50$ in the fourth resolution in our experiments) than the one in \cite{Meyer2020_control_synthesis_ship_docking}  ($50 \times 50 \times 50$), which roughly amounts to the same cell size.
We choose $\tau$ to be 3 seconds. We modified TIRA~\cite{TIRA_Meyer_2019} to generate the full reachable set instead of just the one at $T = \tau = 3$ seconds. Instead, it now results in a sequence of axis-aligned rectangles representing the reachable sets in the different sub-intervals of $[0,\tau]$. In our case study, we choose to partition $[0,\tau]$ to four intervals, resulting in reachable sets with four rectangles. 
 Example reachable sets starting from the origin using the centers of the grid $Q_\inputinst$ over $\inputset$ as constant control signals 
is shown in \ref{fig:redreachableset_9}.


In strategies 1, 2, 4, and 5, we start with batches of size 3 and multiply that by 5 at the end of the for-loop in line~\ref{ln:ssynthesis_quantized_for_over_batches} until the limit is reached (400 or all $9^3$, depends on the strategy). We consider the greedy approach to update $A$ except in the last one, at which we follow the random approach to sample 75 unique symbols to explore.

We ran all experiments on an 18 core Intel i9 7980xe with an AMD64 instruction set and stock speeds with ASUS multicore enhancement defaults, 80 GB of 2666 MHz ddr4 RAM with a 16 GB ZRAM partition, and Ubuntu 22.04 operating system.



	\begin{figure}
  \begin{center}
    \includegraphics[width=\textwidth]{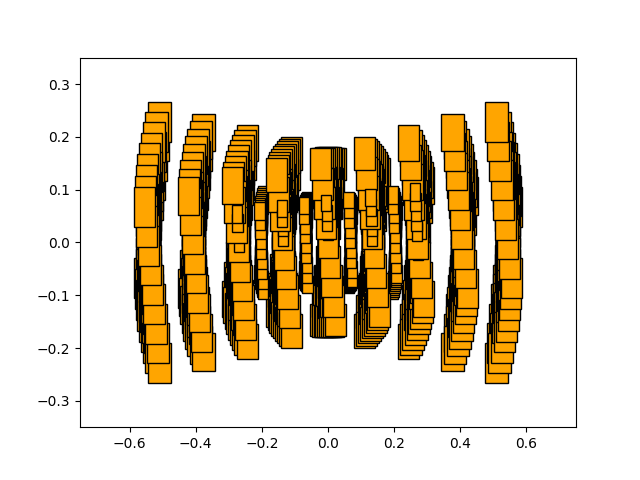}
  \end{center}
  \caption{$\redreachset$: set of reachable sets, projected into the position coordinates, for the ship example having the origin as the initial state following $9^3$ different constant control signals (centers of the cells in $Q_\inputinst$ that partitions $\inputset$ to 9 intervals in each of the three dimensions) for 3 seconds divided into 4 time steps. \protect\label{fig:redreachableset_9}} 
\end{figure}

\end{document}